\def\AAAs          {\ensuremath{\Aut(A\,{\oplus}\,A^*_{})}}
\def\AAs           {\ensuremath{A\,{\oplus}\,A^*_{}}}
\def\AAVect        {\ensuremath{A{\oplus}A\text-\mathrm{vect}}}
\def\As            {A^*_{\phantom|}}
\def\Aut           {\mathrm{Aut}}
\def\AUT           {\mathrm{AUT}}
\def\AVect         {\ensuremath{A\text-\mathrm{vect}}}
\def\be            {\begin{equation}}
\def\bearl         {\begin{array}{l}}
\def\bearll        {\begin{array}{ll}}
\def\Bfi           {S_\text{B}}           

\def\cala          {{\ensuremath{\mathcal A}}}
\def\calc          {{\ensuremath{\mathcal C}}}
\def\cald          {{\ensuremath{\mathcal D}}}

\def\cdo           {\,{\cdot}\,}

\def\cir           {\,{\circ}\,}
\def\Cls           {S_\text{kin}}

\def\complex       {{\ensuremath{\mathbb C}}}
\def\complexx      {{\ensuremath{\mathbb C}^\times_{}}}

\def\deta          {\delta}

\def\ee            {\end{equation}}
\def\eE            {{\rm e}}
\def\eear          {\end{array}}
\def\End           {{\ensuremath{\mathrm{End}}}}

\def\erf           {\eqref }
\def\eq            {\,{=}\,}
\def\findim        {fini\-te-di\-men\-si\-o\-nal}
\def\Fp            {\mathbb F_p}
\def\Fpx           {\mathbb F_p^\times}

\def\hutbeta       {\xi_\beta}
\def\id            {\mbox{\sl id}}
\def\Id            {\mbox{\sl Id}}

\def\idti          {\mbox{\tiny\sl id}}
\def\ii            {{\rm i}}

\def\iN            {\,{\in}\,}

\newcommand\labl[1]{\label{#1}\ee}
\def\le            {\,{\leq}\,}
\def\Mod           {\mbox{\rm -mod}}

\newcommand\nxl[1] {\\[#1mm]}
\newcommand\Nxl[1] {\\[-1.3em]\\[#1mm]}

\def\OAAs          {\ensuremath{\Oq(A{\oplus}A^*_{})}}
\def\OAAsp         {\ensuremath{\Oq(A^{(p)}{\oplus}A^{(p)\,*}_{})}}

\newcommand\ome[2] {\omega_{#1;#2}}

\def\onemat        {\mbox{\small $1\!\!$}1}

\def\Oq            {\mathrm O_q}
\def\ord           {\mathrm{ord}}

\def\pii           {\tau}
\def\qquand        {\qquad{\rm and}\qquad}
\def\R             {{\ensuremath{\mathbb R}}}

\def\rev           {{\mathrm{rev}}}

\def\srr           {\hspace{-.2pt}\reflectbox{$\backslash$}\hspace{-2.5pt}\reflectbox{$\backslash$}\hspace{-.6pt}}
\def\srrad         {\srr_{\!_{\text{ad}}}}

\def\Times         {\,{\times}\,}
\def\To            {\,{\to}\,}
\def\Vect          {\ensuremath{\mathrm{vect}}}

 \newcommand\void[1]{}

\def\Wey           {S_\text{e-m}}
\newcommand\xmapsto[1]{\shortmid\!\xrightarrow{#1}}
\def\Z             {{\ensuremath{\mathcal Z}}}
\def\zet           {{\ensuremath{\mathbb Z}}}

\documentclass[12pt]{article}
\usepackage{latexsym, amsmath, amsthm, amsfonts} 
\usepackage{enumerate, amssymb, xspace, xypic}
\usepackage[all]{xy}
\usepackage[mathscr]{eucal}
\usepackage{graphicx} \usepackage{rotating}
\usepackage{epstopdf,hyperref}
\usepackage{color}

\newtheorem{thm}{Theorem}

\newtheorem{cor}[thm]{Corollary}

\newtheorem{fact}[thm]{Fact}

\theoremstyle{definition}
\newtheorem{Example}[thm]{Example}

\newtheorem{rem}[thm]{Remark}

\newcommand\eqpic[4]{\begin{eqnarray}
                   \begin{picture}(#2,#3){}\end{picture}\nonumber\\
                   \raisebox{-#3pt}{ \begin{picture}(#2,#3) #4 \end{picture} }
                   \label{#1} \\~\nonumber \end{eqnarray} }
\newcommand\Eqpic[4]{\begin{eqnarray}
                   \begin{picture}(#2,#3){}\end{picture}\nonumber\\
                   \raisebox{-#3pt}{ \begin{picture}(#2,#3) #4 \end{picture} }
                   \nonumber \\[3pt]~\label{#1} \end{eqnarray} }
\newcommand\includepic[1]   {{\begin{picture}(0,0)(0,0)
                            \scalebox{.48}{\includegraphics{./
                            pic_fpsv_#1.eps}}\end{picture}}}
\newcommand\Includepic[1]   {{\begin{picture}(0,0)(0,0)
                            \scalebox{.31}{\includegraphics{./
                            pic_fpsv_#1.eps}}\end{picture}}}

\begin{document}

\def\cir{\,{\circ}\,} 
\numberwithin{equation}{section} \numberwithin{thm}{section}
                                     
\begin{flushright}
   {\sf ZMP-HH/14-09}\\
   {\sf Hamburger$\;$Beitr\"age$\;$zur$\;$Mathematik$\;$Nr.$\;$ 509}\\[2mm]
\end{flushright}
\vskip 3.5em
                                   
\begin{center}
\begin{tabular}c \Large\bf 
On the Brauer groups of symmetries 
 \\[3mm] \Large\bf 
of abelian  Dijkgraaf-Witten theories
\end{tabular}\vskip 2.6em
                                     
  ~J\"urgen Fuchs\,$^{\,a}$,~~ Jan Priel\,$^{\,b}$,~~
  ~Christoph Schweigert\,$^{\,b}$~~~and~~~Alessandro Valentino$^{\,b}$

\vskip 9mm

  \it$^a$
  Teoretisk fysik, \ Karlstads Universitet\\
  Universitetsgatan 21, \ S\,--\,651\,88\, Karlstad \\[9pt]
  \it$^b$
  Fachbereich Mathematik, \ Universit\"at Hamburg\\
  Bereich Algebra und Zahlentheorie\\
  Bundesstra\ss e 55, \ D\,--\,20\,146\, Hamburg
                   
\end{center}
                     
\vskip 5.3em

\noindent{\sc Abstract}\\[3pt]
Symmetries of three-dimensional topological field theories are naturally
defined in terms of invertible topological surface defects.
Symmetry groups are thus Brauer-Picard groups.
We present a gauge theoretic realization of all symmetries of abelian 
Dijkgraaf-Witten theories. The symmetry group for a Dijkgraaf-Witten 
theory with gauge group a finite abelian group $A$, and with
vanishing 3-cocycle, is generated by group automorphisms of 
$A$, by automorphisms of the trivial Chern-Simons 2-gerbe on 
the stack of $A$-bundles, and by partial e-m dualities.
 \\
We show that transmission functors naturally extracted from extended 
topological field theories with surface defects give a physical realization 
of the bijection between invertible bimodule categories of a fusion category
$\cala$ and braided auto-equivalences of its Drinfeld center $\Z(\cala)$.
The latter provides the labels for bulk Wilson lines; it follows that a 
symmetry is completely characterized by its action on bulk Wilson lines.

\newpage

\section{Symmetries of abelian Dijkgraaf-Witten theories}\label{intro}
 
Dijkgraaf-Witten theories are extended topological field theories that 
have a mathematically precise gauge theoretic formulation with
finite gauge group. In that setting, the fields of the Dijkgraaf-Witten 
theory with gauge group $G$ are obtained by first considering $G$-bundles, 
to which then in a second step a linearization procedure is applied (see 
\cite{mort6} for a recent description). In the present note we investigate 
the notion of symmetries of three-dimensional Dijkgraaf-Witten theories, 
regarded as extended 1-2-3-dimensional topological field theories. To keep 
the presentation simple we restrict ourselves to the case that the gauge 
group is an abelian group, which we denote by $A$.

\paragraph{Braided auto-equivalences of bulk Wilson lines.}
The task of understanding symmetries in Dijkgraaf-Witten theories can be 
approached from two different angles, either algebraically or gauge 
theoretically. From a purely algebraic point of view, one would consider the 
modular category of bulk Wilson lines, which is the representation category
$\cald(A)\Mod$ of the Drinfeld double of $A$. Symmetries should then in 
particular induce braided auto-equivalences of $\cald(A)\Mod$. 

The group of braided auto-equivalences (up to monoidal natural 
equivalence) can be described as follows. Denote by $A^*_{}$ the group of 
complex characters of $A$.  The group \AAs\ comes with a natural quadratic 
form $q\colon \AAs\To\complexx$, given by $q(g{+}\chi) \eq \chi(g)$ for 
$g\,{+}\,\chi\iN \AAs$. The automorphism group of \AAs\ then has a
subgroup, denoted by \OAAs, consisting of those group automorphisms 
$\varphi$ which preserve this form, i.e.\ satisfy 
$q(\varphi(z) )\eq q(z)$ for all $z\iN \AAs$. Now the group of 
braided auto-equivalences is isomorphic to this group $\OAAs$ 
\cite{enoM}. Simple objects of $\cald(A)\Mod$, and thus simple labels for 
bulk Wilson lines of the Dijkgraaf-Witten theory with gauge group $A$, are in 
bijection with elements of $\AAs$; a braided auto-equivalence induces the 
natural action of the corresponding element of \OAAs\ on the group \AAs.

In this approach the auto-equivalences of $\cald(A)\Mod$ are not intrinsically 
realized in the Dijkgraaf-Witten theory as a gauge theory. It is therefore
not clear whether every braided auto-equivalence of the category of
bulk Wilson lines preserves all
aspects of the three-dimensional topological field theory 
so that it can indeed be regarded as a full-fledged symmetry of the theory.
It is not clear either whether a braided auto-equivalence would then describe
a symmetry uniquely. There might be several different 
realizations, or also none at all, of the auto-equivalences
on other field theoretic quantities, such as boundary conditions.

\paragraph{Universal kinematical symmetries.}
It is thus important to find a field theoretic realization of the
auto-equivalences, relating to the fact that Dijkgraaf-Witten theories can 
be formulated as gauge theories. At the same time we then get additional 
insight into the structure of the group \OAAs. {}From a gauge theoretic point
of view it is natural to expect that the symmetries of the stack 
$\mathrm{Bun}(G)$ of $G$-bundles are symmetries of both classical and 
quantum Dijkgraaf-Witten theories.\,%
 \footnote{~Actually, a general Dijkgraaf-Witten theory involves a
 3-cocycle $\omega\iN Z^3(G,\complexx)$. Here we only consider
 the case of trivial $\omega$, and hence do not expect any
 compatibility relations between the automorphism and $\omega$.}
One might call these symmetries universal kinematical symmetries --
kinematical, because they are symmetries of the kinematical setting, i.e.\
$G$-bundles; and universal, because the manifold on which the $G$-bundles
are defined does not enter.
The symmetries of $\mathrm{Bun}(G)$ form the 2-group $\AUT(G)$, i.e.\ the 
category whose objects are group automorphisms $\varphi\colon G\To G$ and
whose morphisms $\varphi_1\To\varphi_2$ are given by group elements $h\iN G$ 
that satisfy $\varphi_2(g) \eq h\,\varphi_1(g)\,h^{-1}$ for all $g\iN G$.
Since in the case of our interest the group $A$ is abelian, we can safely 
ignore the morphisms in the category $\AUT(A)$ and work with the ordinary 
automorphism group $\Aut(A)$ of the group $A$. 

The group $\Aut(A)$ of symmetries of $\mathrm{Bun}(A)$ can be identified 
in a natural way with a subgroup of the group $\OAAs$ of braided 
auto-equivalences. Indeed, for any $\alpha\iN \Aut(A)$, the automorphism 
$\alpha \,{\oplus}\, (\alpha^{-1})^*_{\phantom*}$ of \AAs\ belongs to 
\OAAs, where $(\alpha)^*_{} \iN \Aut(A^*_{})$ is defined by
$[\alpha^*_{}\chi](a) \,{:=}\, \chi(\alpha(a))$
for all $\chi \iN A^*_{}$ and all $a\iN A$. But this argument is
purely group theoretical, and it is not clear at this point whether
the embedding has any physical relevance and relates symmetries of
bundles to braided auto-equivalences of bulk Wilson lines.

\paragraph{Universal dynamical symmetries.}
The realization of Dijkgraaf-Witten theories as gauge theories 
leads to even more symmetries. Apart from a finite group $G$, a three-cocycle
$\omega\iN Z^3(G,
   \linebreak[0]
U(1))$ is another ingredient of a Dijkgraaf-Witten
theory. Geometrically this cocycle is interpreted \cite{willS}
as a (Chern-Si\-mons) 2-gerbe on the stack $\mathrm{Bun}(G)$ of $G$-bundles,
and we may think of $\omega$ heuristically as a topological Lagrangian.
In the present note we restrict ourselves to the case of vanishing
cocycle $\omega$, corresponding to a trivial 2-gerbe. Still, the
automorphism group of the trivial 2-gerbe is a non-trivial 3-group: 
it is the 3-group of 1-gerbes on $G$. It is thus again natural to 
expect that this 3-group provides us with symmetries of the 
Dijkgraaf-Witten theory with gauge group $G$. We call these symmetries
dynamical universal symmetries, as they involve
symmetries of the topological Lagrangian. 

By the results of \cite{willS}, the objects of the 3-group of 1-gerbes on 
$G$ are 2-cocycles on $G$; isomorphism classes are described by elements of 
the group cohomology $H^2(G,\complexx)$.  The group generated by classical 
kinematical and dynamical symmetries has the structure of a semi-direct 
product, $H^2(G,\complexx)\,{\rtimes}\,\mathrm{Aut}(G)$.
By \cite[Prop.\,4.1]{niRie} this group is isomorphic to the automorphism 
group of the fusion category $G$-\Vect\ of $G$-graded vector spaces.
Indeed, this fusion category enters in the construction of Dijkgraaf-Witten
theories as topological field theories of Turaev-Viro type.

If $G \eq A$ is abelian, cohomology classes in $H^2(A,\complexx)$ are
in bijection with alternating bicharacters. 
(An alternating bicharacter is a map $\beta\colon A\Times A\To \complexx$
that is a group homomorphism in each argument and satisfies
$\beta(a,a) \eq 1$ for all $a\iN A$, and thus 
$\beta(a_1,a_2) \eq \beta(a_2,a_1)^{-1}$ for all $a_1,a_2\iN A$.) 
Again, there is a natural embedding 
$H^2(A,\complexx) \,{\hookrightarrow}\, \OAAs$ of finite groups: to a class
in $H^2(A,\complexx)$ described by an alternating bicharacter
$\beta$, we associate a map $\phi_\beta\colon \AAs\To\AAs$ defined by
$\phi_\beta(a\,{+}\,\chi) \,{:=}\, (a \,{+}\,\beta(a,-){+}\chi(-))$. One
immediately verifies that $\phi_\beta$ is an element of \OAAs. Again
it remains to be shown, though, that this embedding is of physical 
relevance in the sense that it relates symmetries of the topological
Lagrangian to braided auto-equivalences.

\paragraph{Electric-magnetic dualities.}
The universal kinematical and dynamical symmetries cannot, however, 
exhaust the symmetries of Dijkgraaf-Witten models -- the subgroup of \OAAs\ 
generated by them is a proper subgroup. As an illustration, consider 
the case that $A$ is the cyclic group $\zet_2$. This group does not admit 
any non-trivial automorphisms, i.e.\ $\Aut(A) \eq 1$. It does not admit 
any non-trivial alternating bicharacter either, and hence the group
generated by the universal dynamical and kinematical symmetries is trivial. 
On the other hand, the group $\Aut(\AAs)$ is the symmetric group $S_3$ that 
permutes the three order-two elements of \AAs. Its subgroup \OAAs\ is the 
subgroup $S_2 \,{\cong}\, \zet_2$ of $S_3$ whose non-trivial element 
exchanges the generator of $A$ with the one of $A^*_{}$; in physics 
terminology, a transformation
of this type is called an electric-magnetic duality, or e-m duality.  
The presence of such electric-magnetic dualities is a central feature of 
gauge theories in various dimensions (see e.g.\ \cite{kaWi,kbss2} for a 
general discussion). Electric-magnetic dualities have a particularly explicit 
description in theories that can be realized as lattice models, compare
\cite{drwa,kitaA,bcka} and references therein.

\paragraph{Topological surface defects and bimodule categories.}
A proper understanding of the situation, including a physical
realization of the subgroups described above, calls for a unified field 
theoretic perspective. In this note we explain that in the present situation,
for which no rigorous definition of symmetry for an extended topological 
field theory has been fully tested out so far, topological surface 
defects provide such a perspective. In fact, the relation between symmetries 
and classes of invertible topological codimension-one defects has been 
established long ago \cite{ffrs3,ffrs5} for the case of {\em two}-dimensional 
field theories. But the mechanism that implements symmetries via topological 
defects is not restricted to the two-dimensional case. One of the virtues of 
realizing symmetries in terms of topological defects of codimension one is 
that this realization immediately determines
how the symmetries act on all kinds of aspects of the field theory, including
in particular labels of boundaries and defects.

Topological surface defects in 3d TFTs have recently attracted increasing
interest, see \cite{bomb3,kiKon,kaSau3,bamS,ekrs,khtH,gagP2,fusV} for a selection
of recent contributions. The case of three-dimensional topological
field theories of Turaev-Viro type is particularly well understood. In
particular, it is by now well-established that topological surface defects 
in Dijkgraaf-Witten theories with gauge group $A$
correspond to bimodule categories over the 
fusion category $\cala \eq \AVect$ of \findim\ $A$-graded vector spaces. 
Those defects which describe symmetries correspond to invertible
bimodule categories; accordingly we call them invertible defects. Their 
fusion product with the opposite defect is the monoidal unit for fusion,
which is also called the invisible or transparent defect. Invertible 
defects can alternatively be characterized by the fact that the only bulk 
Wilson lines that `condense' on them are the invisible bulk Wilson lines.

The group of (equivalence classes of) invertible bimodule categories, the 
so-called Brauer-Picard group of $\cala$, has been described in \cite{naNi,enoM,daNik}.
In particular, a bijection has been established \cite[Thm\,1.1]{enoM}
between invertible bimodule categories of a fusion category -- in our 
case $\AVect$ -- and braided auto-equivalences of its center -- in our
case the category $\cald(A)\Mod$ of bulk Wilson lines.
As a consequence, also (equivalence classes of) invertible bimodule 
categories are described by the group \OAAs.

\paragraph{The transmission functor.}
The results of \cite{enoM} are of purely representation theoretic
nature. The purpose of the present note is to investigate their 
consequences and counterparts in Dijkgraaf-Witten theories as gauge 
theories. The bijection between equivalence classes of invertible 
bimodule categories and braided auto-equivalences in \cite{enoM}
leads us to consider braided auto-equivalences $F_d$ of $\cald(A)\Mod$
labeled by invertible bimodule categories $d$ over $\cald(A)\Mod$.
Thus to any invertible topological surface defect $d$ we have
to associate such a braided auto-equivalence. 

Now in an extended three-dimensional topological field theory, functors 
are obtained from surfaces with boundaries, and there is indeed a natural 
candidate for the relevant two-dimensi\-o\-nal
cobordism with defect. Namely, to yield an endofunctor of the category of bulk Wilson lines,
the cobordism should have one ingoing and one outgoing boundary; 
and it should not induce any additional topological information; 
hence we have to consider a
cylinder. The cylinder can be thought of as coming from
a cut-and-paste boundary in a three-dimensional topological
field theory. Such boundaries have to intersect surface defects
transversally. Hence a surface defect results in a line embedded in
the cobordism. We are thus lead to  consider a cylinder 
$Z \eq S^1\Times [-1,1]$ with a defect line along the circle 
$D \eq S^1\Times \{0\} \,{\subset}\, Z$, as shown in the following picture:
  \eqpic{01}{200}{28} {
  \put(0,0)    {\includepic{11}
  \put(126,28) {\boldmath{$d$}}
  \put(9.5,-8) {$\scriptstyle -1$}
  \put(105,-8) {$\scriptstyle 0$}
  \put(190,-8) {$\scriptstyle 1$}
  } }
In the sequel we regard the 
circle $S^1\Times \{-1\} \,{\subset}\, Z$ as incoming and the circle 
$S^1\Times \{1\} \,{\subset}\, Z$ as outgoing. We denote the functor described
by the cobordism \erf{01} by $F_d$ and call it the \emph{transmission functor}
for the defect $d$. We will show in section \ref{sec:transmiss}
that for an invertible defect in a general three-dimensional 
extended topological field theory, the transmission functor $F_d$ is
a braided auto-equivalence of the category of bulk Wilson lines.
The transmission functor describes what happens to the type of a bulk 
Wilson line when it passes through the surface defect $d$.

We note that in some physical applications Wilson lines can be interpreted 
as world lines of quasi-particles, with the type of the quasi-particle
specified by the type of the Wilson line. When such a quasi-particle 
crosses an invertible topological surface defect of type $d$, then
the type of quasi-particle is changed according to the transmission
functor $F_d$. In field-theoretic terms, this change is brought about by 
a so-called Alice string \cite{schwa6,abcmw,deBa}. Let us
illustrate this interpretation with the situation that the surface 
defect is a half-plane $\R_{x\geq 0}\Times \R$ in three-dimensional 
space $\R^3 \,{\cong}\, \R^2\Times \R$. The boundary of the half-plane 
consists of a Wilson line that separates the surface defect $d$ from the 
transparent defect (such Wilson lines always exist). The intersection of 
the defect $d$ with a plane $\R^2\Times\{t_0\}$ of fixed time is a 
half-line labeled by $d$; this half-line constitutes the Alice string.
Since the surface defect is topological, the precise position of the 
half-line does not matter.
Whenever a quasi-particle crosses the world surface swept out by the
Alice string, i.e.\ crosses the topological surface defect, it changes
its type according to the transmission functor.

\medskip

In the case of Dijkgraaf-Witten theories, transmission functors are 
explicitly accessible: there is a gauge theoretic realization of
topological surface defects in Dijkgraaf-Witten theories based on 
relative bundles \cite{fusV2}. As a consequence, topological defects 
are classified by a subgroup $H \le A \,{\oplus}\, A$ together with a 
cohomology class in $H^2(H,\complexx)$. The formalism developed in 
\cite{mort6} then allows one to compute the transmission functor.

In this note we provide a set of generators for the group \OAAs\
of braided auto-equi\-valences,
which implies that universal kinematical and dynamical symmetries
together with electric-magnetic dualities generate all symmetries. We then 
give, for each of these generators of \OAAs, a topological defect, compute 
the resulting transmission functor and show that it acts on simple labels 
for bulk Wilson lines by the natural action of \OAAs\ on \AAs.
This provides a field theoretic realization in terms of topological surface 
defects for all braided auto-equivalences. At the same time it establishes 
that the embeddings of the subgroups of dynamical and kinematical
universal symmetries described above are indeed physical. When combined with 
the results of \cite{enoM}, it also follows that the braided equivalences
of bulk Wilson lines are in bijection with field-theoretic symmetries.\,%
 \footnote{~This is reminiscent of the situation in two-dimensional
 rational conformal field theories, where the action of topological line
 defects on bulk fields characterizes isomorphism classes of defects, so that
 the action of topological line defects on bulk fields 
 has been used in the classification of defects.}
Hereby we realize all elements of the
Brauer-Picard group as gauge-theoretic dualities.

\paragraph{Plan of the paper.}
The rest of this note is organized as follows. Section \ref{sec:2} collects some background 
about topological surface defects in Dijkgraaf-Witten theories and 
provides information about transmission functors arising
from invertible defects. In 
Section \ref{sec:3} we construct these defects explicitly for various 
classes of generators and compute their transmission functors.
Finally we show in Section \ref{sec:A} that the group of invertible 
defects is generated by kinematical and dynamical symmetries together 
with e-m dualities. Technically, this is proven as the group-theoretical 
statement that a certain set of elements of the group \OAAs\ 
generates this group.


\section{Surface defects in DW theories and the transmission functor}
\label{sec:2}

\subsection{Surface defects in Dijkgraaf-Witten theories}

A model independent analysis of 
topological surface defects between topological field theories of
Reshetikin-Turaev type has been presented in \cite{fusV}.
We summarize the pertinent aspects of that analysis:
For $\calc$ and $\calc'$ modular tensor categories,
a topological surface defect separating the Reshetikhin-Turaev
theories with bulk Wilson lines labeled by $\calc$ and by $\calc'$,
respectively, exists if and only if the modular category
$\calc \,{\boxtimes}\, (\calc')^\rev_{}$ is braided equivalent to the
Drinfeld center of some fusion category $\cala$;\,%
 \footnote{~Then the classes of $\calc$ and $\calc'$ in the Witt group
 \cite{dmno} of modular tensor categories coincide.}
here $(\calc')^\rev_{}$ is the same monoidal category as $\calc'$, but
with opposite braiding. We call the corresponding braided equivalence functor 
$\calc\,{\boxtimes}\,(\calc')^\rev_{} \,{\xrightarrow{\,\simeq\,}}\, \Z(\cala)$
a trivialization of $\calc \,{\boxtimes}\, (\calc')^\rev_{}$.
If such a trivialization exists, then the 
bicategory of defects is equivalent to the bicategory
of module categories over the fusion category $\cala$.

In the present paper we are interested in the case of defects that
separate a Dijkgraaf-Witten theory based on the
abelian group $A$ from itself. Thus the category of bulk Wilson
lines is already a Drinfeld center, $\calc \eq \calc' \eq \Z(\AVect)$,
and accordingly  there is a distinguished trivialization
  \be
  \calc \,{\boxtimes}\, (\calc')^\rev_{} \xrightarrow{\,\simeq\,}
  \,\Z(\AAVect) \,.
  \ee
The defects of our interest are thus classified by 
module categories over the category of $A\,{\oplus}\,A$-graded vector spaces.

Indecomposable \AAVect-module categories have been classified 
in \cite{ostr5}: they correspond to subgroups $H \le A\,{\oplus}\, A$,
together with a two-cocycle on $H$. 
That they describe surface defects of Dijkgraaf-Witten theories
has been explicitly demonstrated in \cite{fusV2}.


\subsection{The transmission functor}

We want to determine the transmission functor $F_d\colon \calc\To\calc$
for an invertible topological surface defect $d$ described by an 
indecomposable module category over \calc. The physical interpretation of 
the transmission functor $F_d$ for an invertible defect is as follows.
When a bulk Wilson line labeled by an object $U\iN\calc$ passes through 
the surface defect $d$, its label changes to $F_d(U)\iN \calc$.
(Recall that no bulk Wilson lines condense on the defect.)

We now explain why the transmission functor for an
\emph{invertible} surface defect in an extended three-dimensional topological 
field theory has a natural structure of a braided auto-equivalence.
First of all, by composing the transmission functor $F_d$ for a surface defect
$d$ with the transmission functor $F_{\overline d}$ for the opposite defect 
$\overline d$ and invoking fusion of defects, we conclude that 
$F_d\circ F_{\overline d} \eq \Id_\calc \eq  F_{\overline d}\circ F_d$,
so that $F_d$ is indeed an auto-equivalence.
To proceed, it will be convenient to draw the cylinder \erf{01} as an annulus
with an embedded defect line, according to
  \eqpic{03}{70}{27} {
  \put(0,-9)   {\includepic{13}
  \put(67,15.4) {\boldmath{$d$}}
  \put(46.8,38) {$\scriptstyle -1$}
  \put(68.8,38) {$\scriptstyle 0$}
  \put(94.3,38) {$\scriptstyle 1$}
  } }
To discuss monoidality we then have to compare the functors
corresponding to the two `trinion' surfaces shown in the following picture:
  \eqpic{04}{330}{36} {
  \put(0,-8)   {\includepic{14}
  \put(64,26)   {\boldmath{$d$}}
  \put(124,26)  {\boldmath{$d$}}
  }
  \put(190,-8) {\includepic{15}
  \put(124,24)  {\boldmath{$d$}}
  } }
For a general defect, the functors associated to these two trinions are 
not isomorphic and the transmission functor is not monoidal; one rather
obtains monoidal functors between categories of local modules over
braided-commutative algebras in \calc. But if the defect is
invertible, then the functors corresponding to the two trinions
are isomorphic. In fact, a natural isomorphism 
  \be
  \otimes \circ\, (F_d \Times F_d) \Longrightarrow F_d\circ\otimes
  \ee
of functors $\calc\Times\calc {\to}\, \calc$ is
furnished by the following three-manifold with corners and defects:
  \eqpic{02}{110}{65} {
  \put(0,-12)   {\Includepic{12}
  } }
Such a three-manifold with corners is to be read as a span of manifolds 
from the bottom lid to the top lid.
To show that this natural transformation provides a monoidal 
structure on the functor $F_d$,
one needs to check an identity of natural transformations. This identity
follows from the fact that the following two three-manifolds with corners 
and defects are related by a homotopy relative to the boundary:
  \Eqpic{06}{370}{83} {
  \put(-26,-15)  {\Includepic{16}
  }
  \put(214,-15)  {\Includepic{17}
  } }
(This homotopy, restricted to the surface defect, looks like the homotopy 
used in two-dimensional topological field theories to show associativity of
the algebras assigned to circles, but its role is rather different.)
In a similar manner the property that the monoidal structure on $F_d$ is
braided can be deduced from the fact that the following two three-manifolds 
with corners and defects are homotopic as well:
  \Eqpic{08}{350}{115} {
  \put(0,-21)    {\Includepic{18}
  }
  \put(210,-21)  {\Includepic{19}
  } }

\medskip

\begin{rem}
In passing, we mention another physical application: According to
\cite{kaSau3}, surface defects provide an interpretation of the so-called
TFT construction (see \cite{scfr2} for a review) of correlators of 
two-dimensional rational conformal field theories associated with the 
category $\calc$. Thereby a surface defect $d$ in particular determines a 
modular invariant torus partition function $\mathrm Z^d$ of the conformal 
field theory. For an invertible defect $d$ with transmission functor $F_d$,
the resulting torus partition function is
of automorphism type; its coefficient matrix reads
  \be
  \mathrm Z_{ij}^d = \delta_{[U_i^{}],[F_d(U_j^\vee)]} \,,
  \ee
where $\{U_i\}$ is a set of representatives of the isomorphism
classes of simple objects of $\calc$.
\end{rem}


\subsection{Transmission functors for Dijkgraaf-Witten theories}\label{sec:transmiss}

In the case of our interest the modular tensor category \calc\ is the 
representation category of the (untwisted) Drinfeld double $\cald(A)$ of a finite 
abelian group $A$, and a topological surface defect is described by a 
subgroup $H$ of $A\,{\oplus}\, A$ and a two-cocycle on $H$.
To obtain the relevant groupoids of bundles we follow
the prescription of \cite{fusV2} to find the appropriate
relative bundles: For a defect associated to the subgroup
$H \le A\,{\oplus}\, A$ with two-cocycle $\theta\iN Z^2(H,\complexx)$,
the objects of the category of relative bundles consist
of  an $A$-bundle $P_A^\pm$ on each of the two cylinders 
$Z_- \,{:}\, S^1\Times [-1,0]$ and $Z_+ \,{:=}\, S^1\Times [0,1]$,
an $H$-bundle $P_H$ on $D$ and an isomorphism
  \be
  \alpha:\quad \mathrm{Ind}_H^{A\oplus A} P_H \xrightarrow{\,\simeq\,}
  (P_A^+)\big|_D^{} \times (P_A^-)\big|_D^{}
  \labl{IndP=PP}
of $A{\oplus} A$-bundles on $D$. Using that the cylinders $Z_\pm$ are 
homotopic to the circle $D$, one can describe all bundles appearing in
\erf{IndP=PP} by bundles on a circle.
And since $\alpha$ is an isomorphism, one can work with an
equivalent groupoid in which only the $H$-bundles appear as
data. As a consequence the category of relative bundles can be
replaced by the action groupoid $H \srrad H$ for the adjoint action
of $H$ on itself. The objects of this groupoid are group elements $h\iN H$, 
which can be thought of as holonomies of the $H$-bundle on the defect 
circle with respect to some fixed base point; the morphisms of the groupoid
correspond to gauge transformations.

According to the general picture of Dijkgraaf-Witten theories \cite{mort6}, 
for the cylinder we thus get a span of groupoids. For each boundary
circle, we have the category of $A$-bundles on $S^1$, which we replace 
by the equivalent action groupoid $A \srrad A$. The relevant
functor is restriction of bundles to the boundary components.
To describe it, consider the group homomorphisms obtained
from the canonical projections $p_{1,2}$ for $A\,{\oplus}\, A$ to
its two summands,
  \be
  \pi_i: \quad H \hookrightarrow A\,{\oplus}\, A \xrightarrow{\,p_i\,} A \,.
  \ee
These give rise to functors 
  \be
  \hat\pi_i:\quad  H \srrad H \to A \srrad A
  \ee
on action groupoids, acting both on objects and morphisms like $\pi_i$.
We thus can replace the span of groupoids of categories of bundles and
relative bundles by the equivalent span 
  \be
  \xymatrix{
  & H \srrad H \ar^{\hat \pi_2}[dr]\ar_{\hat \pi_1}[dl]& \\
  A \srrad A && A \srrad A
  }
  \ee
of finite groupoids. Next we linearize, i.e.\ for each groupoid consider the
category of functors from the groupoid to $\Vect$. This gives us two pullbacks
  \be
  \hat \pi_i^*: \quad [A \srrad A,\Vect]\to [H \srrad H,\Vect] \,,
  \labl{F1}
as well as pushforwards 
  \be
  \hat \pi_{i\,*}: \quad [H \srrad H,\Vect]\to [A \srrad A,\Vect] 
  \labl{F2}
as their two-sided adjoints. Note that the category
$[A \srrad A,\Vect] \,{\cong}\, \cald(A)\text{-mod} \,{\cong}\, \calc$
is the category of labels for bulk Wilson lines of the Dijkgraaf-Witten theory.

We finally construct a functor $[H \srrad H,\Vect]\To [H \srrad H,\Vect]$
from the two-cocycle $\theta$, following \cite[Sect.\,5.4]{mort6}. To this 
end we first transgress $\theta\iN Z^2(H,\complexx)$ to $\omega_\theta
\iN Z^1(H \srrad H,\complexx)$, a one-cocycle for the loop groupoid 
$H \srrad H\cong[*\srr\zet,*\srr G]$. 
According to \cite[Thm.\,3]{willS} this is the commutator 
  \be
  \omega_\theta(h_1;h_2) = \frac{\theta(h_1,h_2)}{\theta(h_2,h_1)} \,,
  \labl{commcoc}
which is an alternating bicharacter on the abelian group $A$. (As is well 
known, alternating bicharacters for an abelian group $A$ are in bijection 
with the group cohomology $H^2(A,\complexx)$.) The groupoid algebra
$\complex[H\srrad H]$ has as a basis the morphisms 
$b_{\gamma;h}\colon \gamma\,{\xrightarrow{\,h}\,}\, \gamma$ in $H\srrad H$;
its product is composition of morphisms, wherever
this is defined, and zero else. We can canonically identify
  \be
  \complex[H\srrad H]\text{-mod} \simeq [H\srrad H,\Vect] \,.
  \ee
The two-cocycle $\omega_\theta$ gives an algebra automorphism
  \be
  \begin{array}{rll}
  \varphi_\theta:\quad
  \complex[H\srrad H] & \!\!\to \!\! & \complex[H\srrad H] \,, \\[.3em]
  b_{\gamma;h} & \!\!\mapsto\!\! & \omega_\theta(\gamma;h)\, 
  b_{\gamma;h} \,,
  \end{array}
  \ee
which in turn provides us with the desired functor
  \be
  \varphi_\theta^*:\quad
  \complex[H\srrad H]\text{-mod} \to  \complex[H\srrad H]\text{-mod} \,.
  \labl{F3}

The transmission functor $F_{H,\theta}$ is now obtained 
\cite[Sect.\,5.4]{mort6} by pre- and post-composing this functor 
with the pullback and pushforward functors obtained above:
  \be
  F_{H,\theta}:\quad
  [A\srrad A,\Vect] \xrightarrow{\,\hat\pi_1^*\,}
  [H\srrad H,\Vect] \xrightarrow{\,\varphi^*_\theta\,}
  [H\srrad H,\Vect] \xrightarrow{\,{(\hat\pi_2^{})}_*\,}
  [A\srrad A,\Vect] \,.
  \ee
In particular the transmission functor is explicitly computable.
Thus for any given invertible surface defect $(H \le A{\oplus} A, \theta)$ 
of the Dijkgraaf-Witten theory with gauge group $A$
we can find the corresponding braided equivalence $F_{H,\theta}$ explicitly.
{}From these explicit expressions, it is clear that the transmission functor 
only depends on the cohomology class of $\theta$.


\subsection{Action of the transmission functor on simple objects}

Let us determine the action of the transmission functor
on the isomorphism classes of simple objects. For the double
of a general finite group $G$ these classes  are in bijection with
pairs consisting of a conjugacy class $c$ of $G$ and an irreducible 
representation $\chi$ of the centralizer of a representative of $c$. If 
$G \eq A$ is abelian, this reduces to pairs $(a,\chi)$
consisting of a group element $a\iN A$ and an irreducible character 
$\chi\iN \As$; thus the isomorphism classes of simple objects are
  \be 
  \pi_0([A \srrad A,\Vect]) \cong \AAs .
  \ee
The group structure on $\pi_0([A \srrad A,\Vect])$ coming from the monoidal 
structure on the category $[A \srrad A,\Vect]$ coincides with the natural
group structure on \AAs.

It is straightforward to determine the action of each of the three 
functors \erf{F1}, \erf{F2} and \erf{F3} on such pairs.
First, for the pullback along $\hat\pi_1$ we find
  \be
  \begin{array}{rll}
  [\hat\pi_1^*]:\quad 
  \pi_0([A \srrad A,\Vect]) &\!\! \to \!\!& \pi_0([H \srrad H,\Vect]) \,, \Nxl1
  (a\,,\chi) &\!\! \mapsto \!\!& \displaystyle\bigoplus_{h\in p_1^{-1}(a)}
  (h\,,\,p_1^*\chi)
  \eear
  \ee
with $p_1^*$ defined by $[p_1^*\chi](h) \,{:=}\, \chi(p_1(h))$.
Second, the functor $\varphi_\theta^*$ acts as
  \be
  \begin{array}{rll}
  [\varphi_\theta^*]:\quad 
  \pi_0([H \srrad H,\Vect]) &\!\! \to \!\!& \pi_0([H \srrad H,\Vect]) \,, \Nxl1
  (h\,,\psi) &\!\! \mapsto \!\!& \displaystyle
  (h\,,\,\psi {+} \omega_\theta(h;-))
  \eear
  \ee
with $\omega_\theta$ as in \erf{commcoc}.
And third, the pushforward along $\hat \pi_2$ maps
  \be
  \begin{array}{rll}
  [(\hat\pi_2)_*]:\quad 
  \pi_0([H \srrad H,\Vect]) &\!\! \to \!\!& \pi_0([A \srrad A,\Vect]) \,, \Nxl1
  (h\,,\psi) &\!\! \mapsto \!\!& \displaystyle\bigoplus_{\chi_2^{}\in \As}
  (p_2(h)\,, \chi_2)\,\delta_{p_2^*\chi_2^{},\psi} \,. 
  \eear
  \ee


\section{Realizing the symmetries}\label{sec:3}

Various aspects of the group \OAAs\ have been described in \cite[Sect.\,5C]{daNik}. 
Here we present a set of generators of this group: As discussed in detail in
Section \ref{sec:A}, the group \OAAs\ is generated by the following elements:

\begin{enumerate}
\item
The \emph{kinematical universal symmetries}, 
which come from automorphisms of the stack of $A$-bundles. They are
given by the subgroup $\Cls \,{:=}\, 
\big\{\alpha\,{\oplus}\,\alpha^{-1\,*} \,|\, \alpha\iN\Aut(A)\big\}$,
which is isomorphic to $\Aut(A)$. 

\item
The \emph{dynamical universal symmetries}, 
which can be identified with the group of 
(equivalence classes of) 1-gerbes on the stack of $A$-bundles. 
They are given by the group of alternating bicharacters on $A$.
In the terminology of quantum field theory \cite{shWi},
the connection on a 1-gerbe is called a $B$-field. Accordingly we 
refer to the subgroup of alternating bicharacters as B-fields 
and denote it by $\Bfi$.

\item
\emph{Partial electric-magnetic} (or e-m, for short) \emph{dualities}.
Such a symmetry is given by the exchange, in \AAs, of a cyclic summand 
$C$ of $A$ with its character group $C^*_{}$. More explicitly, 
for $A$ written in the form $A \eq A'\,{\oplus}\, C$ with $C$ a 
cyclic subgroup, it acts on 
$\AAs \eq A'\,{\oplus}\, C\oplus (A')^*_{} \,{\oplus}\, C^*$ as
$\id_{A'} \,{\oplus}\, \delta \,{\oplus}\, \id_{(A')^*_{}} \,{\oplus}\,
\delta^{-1}$, with $\delta\colon C\To C^*$ any isomorphism from 
$C$ to $C^*_{}$.
 \\
If one fixes a decomposition of $A$ into a direct sum of cyclic groups
$C_i$, together with an isomorphism $\delta_i\colon C_i \To C_i^*$ 
for each cyclic summand, then the corresponding partial e-m dualities 
generate a subgroup of \OAAs, which we denote by $\Wey$.
\end{enumerate}

\noindent

For each type of generator, we will now specify the subgroup $H$ of 
$A\,{\oplus}\,A$ and
cocycle $\theta$ that label the corresponding invertible surface defect.

\begin{rem}
In principle, for any element of the group \OAAs\
the subgroup $H \le A \,{\oplus}\, A$ and the cocycle 
$\theta$ can be computed from the results in \cite[Sect.\,10.2]{enoM}. 
However, Theorem 1.1. of \cite{enoM} ensures that there is a bijection
between equivalence classes of invertible topological surface defects
and equivalence classes of braided equivalences. Hence it is sufficient
to verify that a given defect described by a pair $(H,\theta)$ reproduces
the correct braided equivalence.
\end{rem}

We will make use of the following fact
(see Corollary 3.6.3 of \cite{davy19} and Proposition 5.2 of \cite{niRie}):

\begin{cor}\label{niRie52}
The \AVect-bimodule category associated with the pair $(H,\theta)$ is 
invertible iff
  \be
  H \cdot (A \oplus \{0\}) = A \oplus A = (A \oplus \{0\}) \cdot H
  \ee
and the restriction of the commutator cocycle $\omega_\theta$ \erf{commcoc} to 
  \be
  H_\cap := \big( H \cap (A \oplus \{0\}) \big) \times
  \big( (A \oplus \{0\}) \cap H \big)
  \ee
is non-degenerate.
\end{cor}


\subsection{Kinematical symmetries: group automorphisms}\label{sec:Cls}

The automorphisms in $\Cls$ are the symmetries of the stack $\mathrm{Bun}(A)$ 
and are thus symmetries of the classical configurations.

A group automorphism $\alpha\colon A\To A$ induces a group automorphism 
$\alpha^*_{}\colon \As \To \As$ acting on $\chi\iN\As$ as
  \be
  [\alpha^*_{}\chi](a) := \chi(\alpha(a))
  \ee
for all $a\iN A$. The combined group automorphism
$\widetilde\alpha \,{:=}\, \alpha \,{\oplus}\, (\alpha^{-1})^*_{\phantom*}
\colon \AAs \To \AAs$ satisfies
  \be
  \bearll
  q\big(\widetilde\alpha(a{+}\chi)\big) \!\!&
  = q\big(\alpha(a)\,{+}\,{\alpha^{-1}}^*_{}(\chi))
  \Nxl3&
  = [ {\alpha^{-1}}^*_{}(\chi)] (\alpha(a))
  = \chi( \alpha^{-1}\alpha(a)) =\chi(a) = q(a{+}\chi) \,,
  \eear
  \ee
i.e.\ preserves the quadratic form $q$ and is thus an element of \OAAs.

We claim that the surface defect whose transmission functor corresponds 
to the automorphism $\widetilde\alpha$ is the following: For the subgroup, we 
take the graph of $\alpha$, i.e.
  \be
  H_\alpha := \{ (a,\alpha(a)) \,|\, a\iN A\} \,<\,  A\oplus A \,,
  \labl{H_alpha}
and for two-cocycle on $H_\alpha$ the trivial two-cocycle $\theta_\circ$.
(We could actually take any exact two-cocycle; for the transmission functor
only the cohomology class matters.)
For instance, for $\alpha \eq \id$, $H$ is the diagonal subgroup of 
$A\,{\oplus}\,A$, which describes the invisible defect, while for the 
`charge conjugation' $a \,{\mapsto}\, a^{-1}$ it is the antidiagonal subgroup.

Let us first check that the pair $(H_\alpha,\theta_\circ)$ defines an
invertible surface defect. We have
  \be
  H_\alpha \cdot (A \oplus \{0\}) = \{ (ab\,, \alpha(a) \,|\, a,b \iN A \}
  = A \oplus A 
  \labl{H_alpha-A0}
and analogously $(A \oplus \{0\}) \cdo H_\alpha \eq A \oplus A$. Moreover,
  \be
  H_\alpha \cap (A \oplus \{0\}) = \{ 0 \} = (A \oplus \{0\}) \cap H_\alpha \,,
  \labl{H_alpha-cap}
so that trivially the restriction of $\omega_{\theta_\circ^{}}$ to 
$(H_\alpha)_\cap^{}$ is non-degenerate. Thus both conditions in Corollary 
\ref{niRie52} are satisfied, and hence the defect labeled by 
$(H_\alpha,\theta_\circ)$ is indeed invertible.

Next we compute the action of the transmission functor 
$F_{H_\alpha,\theta_\circ}$ on isomorphism classes of simple objects. 
The functor $\varphi_{\theta_\circ^{}}^*$ 
is the identity, so that the transmission functor is the composition
  \be
  (a;\chi) \,\xmapsto{~[\hat\pi_1^*]~}\, 
  \left(a,\alpha(a);\tilde \chi\right) \,\xmapsto{~[\hat\pi_{2*}^{}]~}\,
  (\alpha(a),{\alpha^{-1}}^*_{}(\chi))
  \ee
where 
$\tilde\chi \iN H^*_\alpha$ is defined by $\tilde\chi(a,\alpha(a)) \eq \chi(a)$.
Thus indeed $F_{H_\alpha,\theta_\circ}$ acts on isomorphism classes 
of simple objects by $\widetilde\alpha \iN \OAAs$.


\subsection{Dynamical symmetries: B-fields}

These symmetries come from automorphisms
of the trivial 2-gerbe on $\mathrm{Bun}(A)$. They are thus
symmetries of the classical action. 
A dynamical symmetry is described by an alternating bicharacter 
$\beta\colon A\Times A\To\complexx$. To such a bicharacter $\beta$ we 
associate the group homomorphism $\hutbeta\colon A\To A^*$ that acts as
$[\hutbeta(a)](b) \eq \beta(a,b)$ for $a,b\in A$. The automorphism 
$\widetilde\beta$ for a dynamical symmetry is then given by
  \be
  \widetilde\beta(a{+}\chi) = a + \hutbeta(a) + \chi \,. 
  \ee
This is an automorphism because $\hutbeta$ is a group homomorphism, and 
it is in \OAAs\ because $\beta$ is in addition antisymmetric:
  \be
  [\hutbeta(a)](a)=\beta(a,a)= 1
  \ee
for all $a\iN A$, which implies $\beta(a,b) \eq (\beta(b,a))^{-1}$
for $a,b\iN A$.

We claim that the surface defect whose transmission functor
reproduces $\widetilde\beta \iN \OAAs$ looks as follows:
The relevant subgroup is the diagonal subgroup 
$A_\text{diag} \le A\,{\oplus}\,A$ (independently of the
particular choice of $\beta$), and the relevant two-cocycle 
$\theta_\beta$ on $A_\text{diag} \,{\cong}\, A$ is characterized by the 
fact that its commutator cocycle $\omega_{\theta_\beta}$ is $\beta$.
(Recall that for the transmission functor only the cohomology class
of the two-cocycle matters; the alternating
bicharacters are in bijection to these classes.)

Now notice that we have $A_\text{diag} \eq H_{\alpha=\idti}$ with 
$H_\alpha$ as in \erf{H_alpha}, so that as a special case of 
\erf{H_alpha-A0} and of \erf{H_alpha-cap} we see that
  \be
  A_\text{diag} \cdot (A \oplus \{0\}) = A \oplus A
  = (A \oplus \{0\}) \cdot A_\text{diag} 
  \ee
and
  \be
  A_\text{diag} \cap (A \oplus \{0\}) = \{ 0 \} =
  (A \oplus \{0\}) \cap A_\text{diag} \,,
  \ee
respectively. Thus precisely as in Section \ref{sec:Cls} we can conclude 
that the surface defect labeled by the pair $(A_\text{diag},\theta_\beta)$
is invertible.

The action of the transmission functor $F_{\!A_\text{diag},\theta_\beta}$ 
on isomorphism classes of simple objects is obtained as follows:
  \be
  (a;\chi) \,\xmapsto{~[\hat\pi_1^*]~}\, \left(a,a;\tilde \chi\right)
  \,\xmapsto{\,[\varphi^*_{\theta_\deta^{}}]~}\, 
  \left(a,a;\tilde \chi {+} \hutbeta(a) \right)
  \,\xmapsto{~[\hat\pi_{2*}^{}]~}\, \left( a;\chi {+} \hutbeta(a) \right)
  = \widetilde\beta(a{+}\chi) \,,
  \ee
where $\tilde\chi \iN A_\text{diag}$ is defined by $\tilde\chi(a,a) \eq \chi(a)$.
Thus $F_{\!A_\text{diag},\theta_\beta}$ acts on isomorphism classes
precisely by $\widetilde\beta \iN \OAAs$.


\subsection{Partial e-m dualities}

The partial e-m dualities appear as symmetries of quantized Dijkgraaf-Witten 
theories.
Every partial e-m duality can be obtained in the following manner.
Suppose that $A$ is written as a direct sum 
$A \,{\cong}\, A' \,{\oplus}\, C$ with $C$ a cyclic subgroup (allowing
for the case that $A'$ is the trivial subgroup). This induces a similar
decomposition of the character group $\As$: denoting by $C^*_{}$ the 
subgroup of $\As$ of characters that vanish on $A'$, and by $(A')^*_{}$ 
the subgroup of characters vanishing on $C$, we have 
$\As \,{\cong}\, (A')^*_{} \,{\oplus}\; C^*_{}$. 

As abstract groups, $C$ and $C^*$ are isomorphic. Fix an
isomorphism $\deta_C\colon C\To C^*_{}$ and define the automorphism
$\deta$ of the group $\AAs \,{\cong}\, A'\,{\oplus}\, C \,{\oplus}\, 
(A')^*_{} \,{\oplus}\, C^*_{}$ as follows: $\deta$ is the identity on 
the summands $A'$ and $(A')^*_{}$, while on $C \,{\oplus}\, C^*_{}$ it 
acts as
  \be
  (\,c\,,\psi\,) \longmapsto
  \big(\, \deta_C^{-1}(\psi) \,,\deta_C(c) \,\big) \,.
  \ee
That $\deta$ preserves the quadratic form is seen by calculating
  \be
  \bearll
  q\big( \deta(a'{+}c{+}\chi'{+}\psi) \big) \!\!\!&
  = q\big( a'+ \deta_C^{-1}(\psi) + \chi' + \deta_C(c) \big) \Nxl2&
  = \chi'(a')\cdot [\deta_C(c)] (\deta_C^{-1}(\psi))
  = \chi'(a')\cdot \psi(c) = q(a'{+}c{+}\chi'{+}\psi) \,.
  \eear
  \ee

We claim that the surface defect whose transmission functor corresponds 
to $\delta\iN\OAAs$ is as follows: The relevant subgroup 
of $A\,{\oplus}\,A$ is the group 
$H_\deta \,{:=}\, A'_\text{diag}\,{\oplus}\,C\,{\oplus}\, C$, 
where $A'_\text{diag}$ is embedded diagonally into the summand 
$A'\,{\oplus}\, A'$ of $A\,{\oplus}\, A$, while the cocycle 
$\theta_\deta$ on $H_\deta $ is characterized
by its commutator cocycle, which is defined to act as
  \be
  \omega^{}_{\theta_\deta^{}} ((a',c_1,c_2),(\tilde a',\tilde c_1,\tilde c_2))
  := [\deta_C(c_1)](\tilde c_2)
  \cdot \left( [\deta_C(c_2)](\tilde c_1) \right)^{-1}
  \ee
(this is obviously an alternating bicharacter on $H_\deta$). For determining 
the transmission functor, it again suffices to know this bicharacter.

To verify invertibility, note that
  \be
  A'_\text{diag} \cdot (A' \oplus \{0\})
  = \{ (a'b',b') \,|\, a',b' \iN A' \} = A' \oplus A'
  \ee
and analogously 
$(A' \,{\oplus}\, \{0\}) \cdo A'_\text{diag} \eq A' \,{\oplus}\, A'$,
which implies that $(A \,{\oplus}\, \{0\}) \cdo H_\delta \eq A \,{\oplus}\, A
= H_\delta\cdot (A' \oplus \{0\})$. Moreover, we have
  \be
  A'_\text{diag} \cap (A' \oplus \{0\}) = \{0\}
  = (A' \oplus \{0\}) \cap A'_\text{diag} \,,
  \ee
which implies that $(H_\deta)_\cap^{} \eq (C\,{\oplus}\,C)\Times
(C\,{\oplus}\,C)$. To see that $\omega_{\theta_\deta^{}}$ restricted to 
$(H_\deta)_\cap^{}$ is non-de\-ge\-nerate, we fix a generator
$a$ of $C$ and denote by $\psi\iN C^*_{}$ the character with value
$\psi(a) \eq \eE^{2\pi\ii /N}$, with $N \eq |C|$. Then 
$\delta_C(a) \eq \psi^l$ with $l$ such that $(l,N) \eq 1$, and we find
  \be
  \omega^{}_{\theta_\deta^{}}(a^{p_1};a^{q_1},a^{p_2},a^{q_2})
  = \eE^{2\pi\ii l(p_1q_2-q_1p_2)/N} .
  \ee
Now $\eE^{2\pi\ii l/N}$ is a primitive $N$-th root of unity, so that
for any pair $(p_1,q_1)$ we can find $(p_2,q_2)\iN\zet 
   \linebreak[0]
{\times}\,\zet$
such that $p_1q_2 \,{-}\, q_1p_2 \,{\ne}\, 0\bmod N$. Hence 
$\omega_{\theta_\deta^{}}$ is non-degenerate. We can thus again invoke 
Corollary \ref{niRie52} to conclude that the defect labeled by 
$(H_\deta,\theta_\deta)$ is invertible.

To compute the action of the transmission functor on simple
objects, we note that the problem splits into a part involving
only the subgroup $A'$ and another part involving only
the cyclic group $C$. The first problem reduces to the
computation of the transmission functor for the
defect associated with the identity automorphism, which was treated
in section \ref{sec:Cls}. Thus we can restrict ourselves to the case
that $A \eq C$ is cyclic.
In this case the action of the pullback functor on the simple object
$(c,\chi)$ with $b\iN C$ and $\chi\iN C^*_{}$ reads
  \be
  (c;\chi) \,\xmapsto{~[\hat\pi_1^*]~}\, 
  \bigoplus_{\tilde c \in C} \big(c,\tilde c;\chi^{[1]}\big) 
  \labl{sumexpression}
with the character $\chi^{[1]}\iN (C\oplus C)^*_{}$ taking
the values $\chi^{[1]}(d,\tilde d) \eq\chi(d)$ for $d,\tilde d\iN C$.
Next we note that the functor $\varphi^*_{\theta_\deta^{}}$ acts
on simple objects of $\cald(C{\oplus}C)\Mod$ as
  \be
  (c,\tilde c,\chi^{[1]}) \,\xmapsto{\,[\varphi^*_{\theta_\deta^{}}]~}\,
  (c,\tilde c,\chi^{[2]})
  \ee
with the character $\chi^{[2]}\iN (C\oplus C)^*_{}$ taking the values 
  $\chi^{[2]}(d,\tilde d)
  \eq \chi(d)\, {[\delta_C(c)](\tilde d)} / {[\delta_C(\tilde c)](d)}$
  for $d,\tilde d\iN C$.
This is, in turn, mapped by the pushforward functor $[\hat\pi_{2*}^{}]$
to those characters $\chi^{[3]}\iN C^*_{}$ for which
$p_2^*\chi^{[3]} \eq \chi^{[2]}$. This condition amounts to the identity
  \be
  \chi^{[3]}(\tilde d) = \chi^{[2]}(d,\tilde d)
  = \chi(d)\, \frac{[\delta_C(c)](\tilde d)} {[\delta_C(\tilde c)](d)}
  \ee
for all $d,\tilde d\iN C$. Considering the dependence of both sides of
this equality on $\tilde d$ determines $\chi^{[3]} \eq\delta_C(c)$, while
the fact that the dependence on $d$ on the right hand side must be trivial 
shows that we need $\chi(d) \eq [\delta_{C}(\tilde c)](d)$ for all $d\iN C$.
This means that in the summation over $\tilde c$ in \erf{sumexpression} only
the term with $\delta_C(\tilde c) \eq \chi$ survives
the pushforward. We conclude that the composition of the three
functors maps the simple object $(c,\chi)$ to a single simple object,
as befits an equivalence. Concretely,
  \be
  (c,\chi) \,\longmapsto\, (\delta_C^{-1}(\chi),\delta_C(c)) \,,
  \ee
and thus the defect realizes an e-m duality.


\section{Generators of \boldmath{\OAAs}}\label{sec:A}

It remains to be shown that the three types of group elements discussed in
the preceding section -- corresponding to kinematical and dynamical classical
symmetries and to partial e-m dualities -- indeed constitute a set of 
generators for the Brauer-Picard group \OAAs.
To this end we have to show that an arbitrary element of \OAAs\ can be 
expressed as a product of elements in a suitable explicitly specified set of 
generators. This description turns out to be similar to the description of 
symplectic or orthogonal groups over the integers (see e.g.\ \cite{huRei,shWi})
and the proof involves a variant of the Euclidean algorithm similar as in
the proof of Bruhat decompositions (see e.g.\ \cite{reede2}). 
Technical complications arise from the 
need to respect the divisibility properties of the orders of the generators.

We start by introducing pertinent notation.  Any finite abelian group $A$ 
can be presented as $A \eq \bigoplus_{p} A^{(p)}$ with the sum being over 
all primes and $A^{(p)}$ a direct sum of cyclic groups of order a power of $p$.
To analyze the group $A$ we present it in terms of some arbitrary, but fixed, 
ordered family $( a_i \mid i \eq 1,2,...\,,r )$ of generators
such that (writing the group product additively)
  \be
  A^{(p)}
  = \bigoplus_{i=1}^r\, \langle a_i \mid {p^{\ell_i}_{}} a_i \,{=}\, 0 \rangle
  = \bigoplus_{i=1}^r\, \zet_{p^{\ell_i}_{}}^{}
  = \bigoplus_s \big( \zet_{p^{s}_{}}^{} \big)^{\!\oplus n_s}_{}
  \labl{A=oplusZi}
with non-negative integers $n_s$, $r \eq \sum_s n_s$ and $\ell_i$. It will 
be convenient to order the generators such that the powers of $p$ appear in 
ascending order, i.e.\ $\ell_i \le \ell_j$ for $i\,{<}\,j$.
It is easy to see that
  \be
  \Aut(A) = \!\times_{p\;\rm prime} \Aut(A^{(p)})
  \quad \text{~as well as~} \quad
  \OAAs = \!\times_{p\;\rm prime} \OAAsp \,.
  \ee
As a consequence we can, and will, restrict our attention to a single 
prime $p$. By a slight abuse of notation, in the sequel we will just write 
$A$ for $A^{(p)}$.

In terms of the generators, a general group element is a linear combination 
$\sum_{i=1}^r \bar \gamma_i\,a_i$ with 
$\bar \gamma_i \iN \zet \,{\bmod}\, p^{\ell_i}_{}$. In the sequel we freely
replace such classes $\bar\gamma_i$ by representatives $\gamma_i \iN \zet$;
also, we denote by $\gamma_i^{-1} \iN \zet$ a representative of the inverse 
of $\gamma_i$ modulo $p^{\ell_i}_{}$. For the character group $A^*_{}$ we 
choose generators $x_i$ in such a way that $x_i(a_i)$ is a primitive
$p^{\ell_i}$th root of unity while $x_i(a_j) \eq 1$ for $i \,{\ne}\, j$, 
so that the quadratic form $q$ is given by
  \be
  q\big( \mbox{$\sum_{i=1}^r$} (\gamma_i\,a_i\,{+}\,\zeta_i\,x_i) \big)
  = \exp\big( 2\pi\ii \sum_{i=1}^r p^{-\ell_i}_{}\, \gamma_i\,\zeta_i \big) \,,
  \ee
and in particular $\ord(x_i) \eq \ord(a_i)$.
With these conventions, an element $g$ of \OAAs\ (or, for that matter,
of $\End(\AAs)$) is determined by the expressions
  \be
  g(a_i) = \sum_{j=1}^r \big( \alpha_{i,j}^g\,a_j + \xi_{i,j}^g\,x_j \big) \qquand
  g(x_i) = \sum_{j=1}^r \big( \beta_{i,j}^g\,a_j + \eta_{i,j}^g\,x_j \big) 
  \labl{g:alpha-eta}
for $i \eq 1,2,...\,,r$,
with suitable constraints on the coefficients $\alpha_{i,j}^g,\,\xi_{i,j}^g,\,
\beta_{i,j}^g,\,\eta_{i,j}^g \iN \zet$ which, however, we do not need to spell out.

We introduce three subgroups of \AAAs:
  \be
  \bearl
  \Cls := \big\{ \alpha\,{\oplus}\,\alpha^{-1\,*} \,|\, \alpha\iN\Aut(A) \big\}
  \, \cong \Aut(A) \,,
  \Nxl3
  \Bfi := \big\langle b_{i,j} \mid 1 \le i \,{<}\, j \le r \big\rangle 
  \qquand
  \Wey := \bigoplus_{i=1}^r\! D_i  \, \cong \zet_2^{\,\oplus r} .
  \eear
  \labl{Cls-Wey-Bfi}
Here $D_i \,{\cong}\,\zet_2$ is generated by the automorphism $d_i$ that 
exchanges $a_i$ and $x_i$ and leaves all other generators fixed, while $b_{i,j}$
is given by
  \be
  b_{i,j}: \quad \left\{ \bearl a_i \mapsto a_i + x_j \,, \Nxl1
  a_j \mapsto a_j - x_i \,, \Nxl1
  a_k \mapsto a_k \quad \text{for} ~ k\not\in\{i,j\} \,, \Nxl1
  x_k \mapsto x_k \,.
  \eear \right.
  \labl{b_ij}
It is not hard to check that the groups \erf{Cls-Wey-Bfi} are actually
subgroups of $\OAAs \,{<}\, \AAAs$. 
The groups \erf{Cls-Wey-Bfi} describe kinematical universal symmetries
($\Cls$), dynamical symmetries or B-fields ($\Bfi$), and partial 
e-m-dualities associated to the direct sum decomposition \erf{A=oplusZi} 
of $A$ ($\Wey$), respectively.

We will also be interested in two particular types of elements of $\Cls$:
for $i \,{\ne}\,j$ satisfying $\ord(a_i) \eq \ord(a_j)$ we set
  \be
  t_{i,j}: \quad \left\{ \bearl a_j \mapsto a_j - a_i \,, \Nxl1
  a_k \mapsto a_k \quad \text{for} ~ k\ne j \,, \Nxl1
  x_i \mapsto x_i + x_j \,, \Nxl1
  x_k \mapsto x_k \quad \text{for} ~ k\ne i \,, 
  \eear \right.
  \ee
and for $\gamma \,{\ne}\, 0 \bmod p$ and any $j$
  \be
  \ome j\gamma: \quad \left\{ \bearl a_j \mapsto \gamma^{-1}\, a_j \,, \Nxl1
  a_k \mapsto a_k \quad \text{for} ~ k\ne j \,, \Nxl1
  x_j \mapsto \gamma\, x_j  \,, \Nxl1
  x_k \mapsto x_k \quad \text{for} ~ k\ne j \,.
  \eear \right.
  \ee

We further introduce a separate notation for those elements of $\Cls$ that 
act as a transposition on pairs of generators of some fixed order and leave 
all other generators fixed, according to
  \be
  \pii_{i,j}: \quad \left\{ \bearl a_i \leftrightarrow a_j \,, \Nxl1
  a_k \mapsto a_k \quad \text{for} ~ k\not\in\{i,j\} \,, \Nxl1
  x_i \leftrightarrow x_j \,, \Nxl1
  x_k \mapsto x_k \quad \text{for} ~ k\not\in\{i,j\} 
  \eear \right.
  \labl{pii}
with $\ord(a_j) \eq \ord(a_i)$. These generate a subgroup 
$\mathfrak S \eq \bigoplus_s \mathfrak S_{n_s} \le \Cls$ consisting of 
elements that permute pairs $(a_i,x_i)$ of generators of the same order.
Below, for convenience we allow for $i\eq j$ in \erf{pii},
i.e.\ for any $i$, $\pii_{i,i}$ is just the unit element of \OAAs.

\medskip

We now establish the following

\begin{fact}
\OAAs\ is generated by the subgroups \erf{Cls-Wey-Bfi}.
\end{fact}

\begin{proof}~
\nxl1
\underline{Step 1}\,:
Given $g\iN\OAAs$ we show that multiplying $g$ with suitable
elements of the subgroups \erf{Cls-Wey-Bfi} yields a group element 
that leaves the last generator $x_r$ invariant.
\nxl2
\underline{Step 1a}\,:
Describe $g$ as in \erf{g:alpha-eta}. We first consider the case that 
$\ord(\eta_{r,i}^g\,x_i) \,{<}\, \ord(x_r)$ for all $i$.
Then, since $g$ must preserve the order of $x_r$, there exists at least one 
value of $i$ such that $\ord(\beta_{r,i}^g\,a_i) \eq \ord(x_r)$.
Take one such value (say, the largest one satisfying the equality) and 
denote it by $k(r)$ or, for brevity, just by $k$. Then the group element 
$g' \,{:=}\, d_k \cir g$ acts as
  \be
  g'(x_r) \equiv d_k(g(x_r)) = g(x_r)
  + (\beta_{r,k}^g \,{-}\, \eta_{r,k}^g)\, (x_r \,{-}\,a_r) \,,
  \ee
so that in particular 
$\ord(\eta_{r,k}^{g'}\,x_k) \eq \ord(\beta_{r,k}^{g'}\,x_k) \eq \ord(x_r)$.
\nxl2
\underline{Step 1b}\,:
By step 1a we can assume that $g$ satisfies 
$\ord(\eta_{r,k}^g\,x_k) \eq \ord(x_r)$, i.e.\ $\ord(x_k) \eq \ord(x_r)$
and $\eta_{r,k}^g \,{\ne}\, 0 \bmod p$. It follows that 
$\pii_{k,r} \iN \mathfrak S \le \Cls$ and that there exists a $\gamma \iN \zet$ 
such that $\gamma\,\eta_{r,k}^g \eq 1 \bmod p$.
Then the group element $g' \,{:=}\, \ome r\gamma \cir \pii_{r,k} \cir g$ 
acts as in \erf{g:alpha-eta} with $\eta_{r,r}^{g'} \eq 1$.
\nxl2
\underline{Step 1c}\,:
Invoking step 1b we assume from now on that $g$ satisfies $\eta_{r,r}^g \eq 1$.
Further, for $i \,{\ne}\, r$ the element 
$g' \,{:=}\, {(b_{i,r})}_{}^{\beta_{r,i}^g}$ satisfies 
$\beta_{r,i}^{g'} \eq \beta_{r,i}^g \,{-}\, \beta_{r,i}^g \eq 0$. Hence by 
composing $g$ successively, for all $i\eq 1,2,...\,,r{-}1$, with
the group element $b_{i,r}$ raised to the power $\beta_{r,i}^g$ one obtains
a group element $\tilde g$ satisfying
  \be
  \tilde g(x_r) = \beta_{r,r}^g\, a_r + x_r + \sum_{i=1}^{r-1}\xi_{r,i}^g\,x_r \,.
  \ee
Now by construction, $\tilde g \iN \OAAs$, while on the other hand
$q(\tilde g(x_r)) \eq \exp(2\pi\ii\,p^{-\ell_r}_{} \beta_{r,r}^g)$. Thus 
in fact we must have $\beta_{r,r}^g \eq 0\bmod p^{\ell_r}$.
\nxl2
\underline{Step 1d}\,:
By step 1c we can assume that $g$ satisfies $\beta_{r,i}^g \eq 0$
for all $i \eq 1,2,,...\,,r$. Further, for $i \,{\ne}\, r$ the element 
$g' \,{:=}\, {(t_{r,i})}_{}^{\eta_{r,i}^g}$ satisfies 
$\eta_{r,i}^{g'} \eq \eta_{r,i}^g \,{-}\, \eta_{r,i}^g \eq 0$. Hence by 
composing $g$ successively, for all $i\eq 1,2,...\,,r{-}1$, with
the group element $t_{r,i}$ raised to the power $\eta_{r,i}^g$ one obtains
a group element $\tilde g$ satisfying $\tilde g(x_r) \eq x_r$.
\nxl3
\underline{Step 2}\,:
By step 1 we can assume that $g(x_r) \eq x_r$. Now consider the image
$g(x_{r-1})$ of the group element $x_{r-1}$. We manipulate it in full analogy 
with what we did with $g(x_r)$ in step 1, just replacing
$r \,{\mapsto}\, r{-}1$ everywhere, but with the following amendment: In case
that in the analogue of step 1a the label $k \eq k(r{-}1)$ should turn out to
take the value $r$, before proceeding to replacing 
$g \,{\mapsto}\, d_k \cir g$ we consider instead of $g$ the group element 
  \be
  g' := t_{r-1,r} \circ g \,.
  \labl{g'=tg}
After this replacement we can assume that $k \le r{-}1$. As a
consequence, afterwards one never will have to compose with
elements from \erf{Cls-Wey-Bfi} of the form $\ome r\gamma$ or 
$b_{r,j} \cir d_r$ which would potentially alter the input relation
$g(x_r) \eq x_r$. Thus by further proceeding along the lines of step 1
we end up with a group element $\tilde g$ satisfying both
$\tilde g(x_r) \eq x_r$ and $\tilde g(x_{r-1}) \eq x_{r-1}$.
\nxl3
\underline{Steps $3,\,4,...\,,r$}\,:
Proceed iteratively for $g(x_{r-j})$ for $j \eq 2,3,...\,,r{-}1$, where in 
the $j$th iteration the role of $t_{r-1,r}$ in \erf{g'=tg} is taken over
by $t_{r-j,r-l}$ for suitable $l \,{<}\, j$.
\\
The result is a group element $\tilde g$ satisfying $\tilde g(x_i) \eq x_i$
for all $i \eq 1,2,...\,,r$.
\nxl3
\underline{Step $r{+}1$}\,:
By step $r$ we can assume that $g(x_i) \eq x_i$ for all $i \eq 1,2,...\,,r$.
We show that this in fact implies that 
$g(a_i) \eq a_i + \sum_{j=1}^r \xi_{i,j}^g\,x_j$ for all $i \eq 1,2,...\,,r$.
Indeed, from \cite{hiRh} we know that in order for $g$ to belong to \AAAs,
the matrix $M(g) \eq \left( \begin{array}{cc}\!\!
\alpha^g \!&\! \xi^g \\[-2pt] \beta^g \!&\! \eta^g \!\!\eear\right)$ with block 
matrices $\alpha^g,\,\xi^g,\,\beta^g,\,\eta^g$ consisting of the coefficients 
in \erf{g:alpha-eta}, must satisfy $\det(M(g) \bmod p) \,{\ne}\,0$. 
\\
Now for $g$ of the form considered here we have $\eta^g \eq \onemat_{r\times r}$ 
and $\beta^g_{} \eq 0$; this implies in particular that
$0 \,{\ne}\, \det(M(g) \bmod p) \eq \det(\alpha^g \bmod p)$, and thus that
$\alpha^g \iN \Aut(A)$. As a consequence, together with $g$ also the product
$g' \,{:=}\, g \cir \big( (\alpha^g)^{-1}_{} {\oplus} (\alpha^g)^*_{}\big)$ 
is an element of \OAAs. On the other hand, we have explicitly
  \be
  g'(a_i) = a_i + \sum_j \xi^{g'}_{i,j}\,x_j  \qquand
  g'(x_i) = \sum_j \eta^{g'}_{i,j}\,x_j \,.
  \ee
Hence the fact that $g'$ belongs to \OAAs\ amounts in particular to the following 
restrictions, which together are also sufficient:
  \be
  \begin{array}{lll}
  \multicolumn3l{ q(g'(a_i)) = q(a_i) ~~\Longrightarrow~~ \xi^{g'}_{i,i} = 0 \,, }
  \Nxl2
  q(g'(a_i{+}a_j)) = q(a_i{+}a_j) & \Longrightarrow &
  \xi^{g'}_{i,j} + \xi^{g'}_{j,i} = 0 ~~\text{for}~~ i \,{\ne}\, j \,,
  \Nxl2
  q(g'(a_i{+}x_i)) = q(a_i{+}x_i) & \Longrightarrow & \eta^{g'}_{i,i} = 1 \,,
  \Nxl2
  q(g'(a_j{+}x_i)) = q(a_j{+}x_i) & \Longrightarrow & 
  \eta^{g'}_{i,j} = 0 ~~\text{for}~~ i \,{\ne}\, j \,.
  \eear
  \ee
Together, these restrictions just say that $g' \iN \Bfi$.
\nxl1
This concludes the proof.
\end{proof}

\bigskip

In the following example we illustrate how the result follows from an explicit
analysis in a particularly simple case.

\begin{Example} $A \eq \zet_p$.
\nxl1
It is not hard to see that in order for \erf{g:alpha-eta} to be in \OAAs, it 
is necessary and sufficient that the numbers $\alpha,\,\beta,\,\xi,\,\eta$ satisfy
$\alpha\xi \eq 0 \eq \beta\eta$ and $\alpha\eta\,{+}\,\beta\xi \eq 1$ modulo $p$.
These constraints are solved by 
  \be
  \xi = 0 = \beta\,,~~ \eta = \alpha^{-1} \qquad\text{and by}\qquad
  \alpha = 0 = \eta\,,~~ \beta = \xi^{-1} .
  \ee
Among the solutions of the second type is in particular the case 
$\xi \eq \beta \eq 1$, which gives the (unique) e-m duality, while all other
solutions of this type are obtained from one of the first type by composing
with the e-m duality. In short, we have
  \be
  \Oq(\zet_p^{}{\oplus}\zet_p^*) = \Cls \rtimes \Wey 
  \ee
with
  \be
  \Cls = \Aut(\zet_p) = GL_1(\Fp) = \Fpx \cong \zet_{p-1} \qquand
  \Wey = \zet_2 \,.
  \ee
In particular, $|\Oq(\zet_p^{}{\oplus}\zet_p^*)| \eq 2(p{-}1)$.

The transmission functor for the automorphism given by $\alpha \,{\in}\, \Fpx$ is
  \be
  \begin{array}{rll}
  \zet_p \oplus \zet_p^* & \!\!\to \!\! & \zet_p \oplus \zet_p^* \,, \\[.3em]
  (a,\chi) & \!\!\mapsto\!\! & (\alpha a, \chi(\alpha^{-1} \cdot))\,
  \end{array}
  \ee
\noindent
and the transmission functor corresponding to the single non-trivial 
e-m duality is given by
  \be
  \begin{array}{rll}
  \zet_p \oplus \zet_p^* & \!\!\to \!\! & \zet_p \oplus \zet_p^* \,, \\[.3em]
  (a,\chi) & \!\!\mapsto\!\! & (\delta(\chi)^{-1}, \delta(a))\, .
  \end{array}
  \ee
\end{Example}

\begin{Example} $A \eq \zet_p^{\,2}$.
\nxl1
As another example consider the non-cyclic group $A=\zet_p\,{\oplus}\,\zet_p$.
In this case the second co\-ho\-mo\-logy group $H^2(A,\complex^\times)$ is
non-trivial. A choice $(e_1,e_2)$ of generators of $A$ allows us to describe 
the subgroups generating the symmetry group $\Oq(\zet_p^{2}\,{\oplus}\,\zet_p^{*2})$
explicitly: 
  \be
  \Cls = \Aut(\zet_p^2) \cong GL_2(\Fp) \,, \qquad
  \Wey \cong \zet_2^2 \,, \qquad
  \Bfi \cong \zet_p \,.
  \ee
Transmission functors for the kinematical symmetries and e-m dualities are similar 
to the previous example. An alternating bicharacter $\beta$ can be described by
  \be
  \beta_\ell(e_i,e_i)=1 \qquand
  \beta_\ell(e_1,e_2)=\beta_\ell(e_2,e_1)^{-1}=\eE^{2\pi\ii\ell/p }
  \ee
with $\ell \iN \{ 0,1,...\,, p{-}1 \}$. This yields a map
$b_\ell\colon \zet_p^{} \,{\to}\, \zet_p^*$ with $b(a)(b) \eq \exp(2\pi\ii \ell ab/p)$.
We then have the following transmission functor:
  \be
  \begin{array}{rll}
  \zet_p^2 \oplus \zet_p^{*2} & \!\!\to \!\! & \zet_p^{2} \oplus \zet_p^{*2} \,, \\[.3em]
  (a_1,a_2,\chi_1,\chi_2) & \!\!\mapsto\!\! & (a_1,a_2, \chi_1{+}b(a_2), \chi_2{-}b(a_1)) \,.
  \end{array}
  \ee
\end{Example}

 \vskip 4em 

{\small \noindent{\sc Acknowledgments:}
JF is supported by VR under project no.\ 621-2013-4207. CS is partially supported 
by the Collaborative Research Centre 676 ``Particles, Strings and the Early 
Universe - the Structure of Matter and Space-Time'' and by the DFG Priority 
Programme 1388 ``Representation Theory''. JF is grateful to Hamburg 
University, and in particular to CS, Astrid D\"orh\"ofer and Eva Kuhlmann, 
for their hospitality when part of this work was done.
JF, CS and AV are grateful to the Erwin-Schr\"odinger-Institute (ESI) for the
hospitality during the program ``Modern trends in topological field theory'' 
while part of this work was done; this was also supported by the network 
``Interactions of Low-Dimensional Topology and Geometry with Mathematical 
Physics'' (ITGP) of the European Science Foundation.
}

 \newpage

 \newcommand\wb{\,\linebreak[0]} \def\wB {$\,$\wb}
 \newcommand\Bi[2]    {\bibitem[#2]{#1}} 
 \newcommand\Epub[2]  {{\em #2} {[\tt #1]}}
 \newcommand\inBO[9]  {{\em #9}, in:\ {\em #1}, {#2}\ ({#3}, {#4} {#5}), p.\ {#6--#7} {\tt [#8]}}
 \newcommand\J[7]     {{\em #7}, {#1} {#2} ({#3}) {#4--#5} {{\tt [#6]}}}
 \newcommand\JJ[7]    {{\em #7}, {#1} {} ({#3}) {#4} {{\tt [#6]}}}
 \newcommand\JO[6]    {{\em #6}, {#1} {#2} ({#3}) {#4--#5} }
 \newcommand\JP[7]    {{\em #7}, {#1} ({#3}) {{\tt [#6]}}}
 \newcommand\Prep[2]  {{\em #2}, preprint {\tt #1}}
 \def\aagt  {Alg.\wB\&\wB Geom.\wb Topol.}     
 \def\adma  {Adv.\wb Math.}
 \def\alnt  {Algebra\wB\&\wB Number\wB Theory}          
 \def\ammm  {Amer.\wb Math.\wB Monthly}
 \def\anop  {Ann.\wb Phys.}
 \def\cntp  {Com\-mun.\wB Number\wB Theory\wB Phys.}
 \def\comp  {Com\-mun.\wb Math.\wb Phys.}
 \def\imrn  {Int.\wb Math.\wb Res.\wb Notices}
 \def\jhep  {J.\wb High\wB Energy\wB Phys.}
 \def\jhrs  {J.\wB Homotopy\wB Relat.\wB Struct.}
 \def\joal  {J.\wB Al\-ge\-bra}
 \def\jram  {J.\wB rei\-ne\wB an\-gew.\wb Math.}
 \def\nupb  {Nucl.\wb Phys.\ B}
 \def\quto  {Quantum Topology}
 \def\phrb  {Phys.\wb Rev.\ B}
 \def\phrl  {Phys.\wb Rev.\wb Lett.}
 \def\tams  {Trans.\wb Amer.\wb Math.\wb Soc.}

\end{document}